\documentclass[11pt]{article} 

\usepackage{graphicx,color}
\usepackage{bm}
\usepackage{epsfig, amssymb, amsmath}
\usepackage{amsfonts}
\usepackage{booktabs}
\usepackage{subfigure}
\usepackage{tikz}
\usepackage{amsthm}
\usetikzlibrary{arrows} 

\def\ga{{\mathcal{G}}} 
\def\outcome{{\mathcal{C}}}

\def\pne{{\sf N}} 
\def\sne[#1]{{#1}{\sf -SN}} 
\def\core{{\sf CR}} 
\def\score{{\sf SCR}} 
\def\poa{{\sf PoA}}
\def\pos{{\sf PoS}}
\def\spoa[#1]{{#1}{\sf -SPoA}}
\def\cpoa{{\sf CPoA}}
\def\spos[#1]{{#1}{\sf -SPoS}}
\def\cpos{{\sf CPoS}}
\def\SW{{\sf SW}}
\def\bbar(#1){\bar{\bar{#1}}}

\newtheorem{theorem}{Theorem}
\newtheorem{lemma}[theorem]{Lemma}
\newtheorem{proposition}[theorem]{Proposition}
\newtheorem{corollary}[theorem]{Corollary}

\allowdisplaybreaks
\begin{document}

\title{Stable Outcomes in Modified Fractional Hedonic Games}  



\author{Gianpiero Monaco,$^{1}$ Luca Moscardelli,$^{2}$ Yllka Velaj$^{2}$\\
\normalsize{$^{1}$University of L'Aquila,}\\
\normalsize{$^{2}$University of Chieti-Pescara}
}

\date{}


\maketitle

\begin{abstract}  
In \emph{coalition formation games} self-organized coalitions are created as a result of the strategic interactions of independent agents. 
For each couple of agents $(i,j)$, weight $w_{i,j}=w_{j,i}$ reflects how much agents $i$ and $j$ benefit from belonging to the same coalition.
We consider the {\em modified fractional hedonic game}, that is a coalition formation game in which agents' utilities are such that the total benefit of agent $i$ belonging to a coalition (given by the sum of $w_{i,j}$ over all other agents $j$ belonging to the same coalition) is averaged over all the other members of that coalition, i.e., excluding herself.
Modified fractional hedonic games constitute a class of succinctly representable hedonic games. 

We are interested in the scenario in which agents, individually or jointly, choose to form a new coalition or to join an existing one, until a stable outcome is reached. To this aim, we consider common stability notions, leading to strong Nash stable outcomes, Nash stable outcomes or core stable outcomes: we study their existence, complexity and performance, both in the case of general weights and in the case of 0-1 weights. 
In particular, we completely characterize the existence of the considered stable outcomes and show many tight or asymptotically tight results on the performance of these natural stable outcomes for modified fractional hedonic games, also highlighting the differences with respect to the model of fractional hedonic games, in which the total benefit of an agent in a coalition is averaged over all members of that coalition, i.e., including herself.  
\end{abstract}



\section{Introduction}
Teamwork, clustering and coalition formations have been important and
widely investigated issues in computer science research. In fact, in
many economic, social and political situations, individuals carry
out activities in groups rather than by themselves. In these scenarios,
it is of crucial importance to consider the satisfaction of the members of the groups. 

Hedonic games, introduced in \cite{DG1980}, model the formation of coalitions of agents.  
They are games in which agents have preferences over the set of all
possible agent coalitions, and the utility of an agent depends on the composition of the coalition she belongs to. While the standard model of hedonic games assumes that agents' preferences over coalitions are ordinal, there are several prominent classes of hedonic games where agents assign cardinal utilities to coalitions. 
Additively separable hedonic games constitute a natural and succinctly representable class of hedonic games.   
In such setting each agent has a value for any other agent, and the utility of a coalition to a particular agent
is simply the sum of the values she assigns to the members
of her coalition. Additive separability satisfies a number of
desirable axiomatic properties \cite{ABS11} and is the non-transferable utility generalization of
graph games studied in \cite{DP94}. Fractional hedonic games, introduced in \cite{ABH2014}, are similar to additively separable ones, with the difference that the utility of each agent is divided by the size of her coalition. Arguably, it is more natural to compute
the average value of all other members of the coalition \cite{EFF2016}. Various solution concepts, such as the core, the
strict core, and various kinds of individual stability like Nash Equilibrium have been proposed to analyze these games (see the Related Work subsection).

In this paper we deal with \emph{modified fractional hedonic games} (MFHGs), introduced in \cite{O12}, and afterward studied in \cite{EFF2016,KKP16}. 
MFHGs model natural behavioral dynamics in social environments. Even when defined on undirected and unweighted graphs, they suitably model a basic economic scenario 
referred to in \cite{ABH2014,BFFMM15} as Bakers and Millers. 
Moreover, MFHGs can model other realistic scenarios: (i) politicians may want to be in a party that maximizes the fraction of like-minded members; (ii) people may want to be with an as large as possible fraction of people of the same ethnic or social group.

In MFHGs, slightly differently than in fractional hedonic games, the utility of an agent $i$ is divided by the size of the coalition she belongs to minus 1, that indeed corresponds to the average value of all other members than $i$ of the coalition. Despite such small difference, we will show that natural stable outcomes in MFHGs perform differently than in fractional hedonic games. Specifically, we adopt \emph{Nash stable}, \emph{Strong Nash stable} and \emph{core} outcomes. Informally, an outcome is Nash stable (or it is a Nash equilibrium) if no agent can improve her utility by unilaterally changing her own coalition. Moreover, an outcome is strong Nash stable if no subset of agents can cooperatively deviate in a way that benefits all of them. Finally, an outcome is in the core or is core stable, if there is no subset of agents $T$, whose members all prefer $T$ with respect to the coalition in the outcome. 
We point out that, (strong) Nash stable outcomes are resilient to a group of agents that can join any coalition and therefore represent a powerful solution concept. However, there are settings in which it is not allowed for one or more agents to join an existent coalition without asking for permission to its  members: in these settings the notion of core, where in a non-stable outcome a subset of $T$ agents can only form a new coalition itself and cannot join an already non-empty coalition, appears to be more realistic.

Our aim is to study the existence, performance and computability of natural stable outcomes for MFHGs. In particular, we evaluate the performance of Nash, strong Nash, and core stable outcomes for MFHGs, by means of the widely used notions of price of anarchy (resp. strong price of anarchy and core price of anarchy), and price of stability (resp. strong price  of stability and core price of stability), which are defined as the ratio between the social optimal value and the social value of the worst (resp. best) stable outcome. 

An instance of MFHG can be effectively modeled by means of a weighted undirected graph $G=(N,E,w)$, where nodes in $N$ represent the agents, and the weight $w(\{i,j\})$ of an edge $\{i,j\} \in E$ represents how much agents $i$ and $j$ benefit from belonging to the same coalition.

\subsection{Related Work}
To the best of our knowledge, only few papers dealt with stable outcomes for MFHGs. Olsen \cite{O12} considers unweighted undirected graphs and investigates computational issues concerning the problem of computing a Nash stable outcome different than the trivial one where all the agents are in the same coalition. The author proves that the problem is NP-hard when we require that a coalition must contain a given subset of the agents, and that it is polynomial solvable for any connected graph containing
at least four nodes. Kaklamanis et al. \cite{KKP16} show that the price of stability is $1$ for unweighted graphs. Finally, Elkind et al. \cite{EFF2016} study the set of Pareto optimal outcomes for MFHGs.

Fractional hedonic games have been introduced by Aziz et al. \cite{ABH2014}. They prove that the core can be empty for games played on general graphs and that it is not empty for games played on some classes of undirected and unweighted graphs (that is, graphs with degree at most $2$, multipartite complete graphs, bipartite graphs admitting a perfect matching and regular bipartite graphs). Brandl et al. \cite{BBS15}, study the existence of core and individual stability in fractional hedonic games and the computational complexity of deciding whether a core and individual stable partition exists in a given fractional hedonic game. Bil{\`{o}} et al. \cite{BFFMM14} initiated the study of Nash stable outcomes for fractional hedonic games and study their existence, complexity and performance for general and
specific graph topologies. In particular they show that the price of anarchy is $\Theta(n)$, and that for unweighted graphs, the problem of computing a Nash stable outcome of maximum social welfare is NP-hard, as well as the problem of computing an optimal (not necessarily stable) outcome. Furthermore, the same authors in \cite{BFFMM15} consider unweighted undirected graphs and show that $2$-Strong Nash outcomes, that is, an outcome such that no pair of agents can improve their utility by simultaneously changing their own coalition, are not always guaranteed. They also provide upper and lower bounds on the price of stability for games played on different unweighted graphs topologies. 
Finally, Aziz et al. \cite{AGGMT15} consider the computational complexity of computing welfare maximizing partitions (not necessarily Nash stable) for fractional hedonic games. 
We point out that fractional hedonic games played on unweighted undirected graphs model realistic economic scenarios referred to in \cite{ABH2014,BFFMM15} as Bakers and Millers.

Hedonic games have been introduced by Dr{\'e}ze and Greenberg \cite{DG1980}, who analyzed them under a cooperative perspective. Properties guaranteeing the existence of core allocations for games with additively separable utility have been studied by Banerjee, Konishi and S{\"o}nmez \cite{BKS01}, while Bogomolnaia and Jackson \cite{BJ02} deal with several forms of stable outcomes like the core, Nash and individual stability. Ballester \cite{B04} considers computational complexity issues related to hedonic games, and show that the core and the Nash stable outcomes have corresponding NP-complete decision problems for a variety of situations, while Aziz et al. \cite{ABS11} study the computational complexity of stable coalitions in additively separable hedonic games. Moreover, Olsen \cite{O09} proves that the problem of deciding whether a Nash stable coalitions exists in an additively separable hedonic game is NP-complete, as well as the one of deciding whether a non-trivial Nash stable coalitions exists in an additively separable hedonic game with non-negative and symmetric preferences (i.e., unweighted undirected graphs). 

Feldman et al. \cite{FLN15} investigate some interesting subclasses of hedonic games from a non-cooperative point of view, by characterizing Nash equilibria and providing upper and lower bounds on both the price of stability and the price of anarchy. It is worth noticing that in their model they do not have an underlying graph, but agents lie in a metric space with a distance function modeling their distance or ``similarity''. Peters \cite{P16} considers ``graphical'' hedonic games where agents form the vertices of an undirected graph, and each agent's utility function only depends on the actions taken by her neighbors (with general value functions). It is proved that, when agent graphs have bounded treewidth and bounded degrees, the problem of finding stable solutions, i.e., Nash equilibria, can be efficiently solved. Finally, hedonic games have also been considered by Charikar et al. \cite{CGW05} and by Demaine et al. \cite{DEFI06} from a classical optimization point of view (i.e, without requiring stability for the solutions) and by Flammini et al. in an online setting \cite{FMMSZ18}.

Peters et al. \cite{PE15} consider several classes of hedonic games and identify simple conditions on expressivity that are sufficient for the problem of checking whether a given game admits a stable outcome to be computationally hard. 

From a different perspective, strategyproof mechanisms for additively separable hedonic and fractional hedonic games have been proposed in \cite{FMZ17,WV2015}.

Finally, hedonic games are being widely investigated also under different utility definitions: For instance, in \cite{BFMO17,BFO17}, coalition formation games, in which agent utilities are proportional to their harmonic centralities in the respective coalitions, are considered.

\subsection{Our Results}
We start by dealing with strong Nash stable outcomes. We first prove that there exists a simple star graph with positive edge weights that admits no strong Nash stable outcomes. Therefore we focus on unweighted graphs, and present a polynomial time algorithm that computes an optimum outcome that can be transformed in a strong Nash stable one with the same social welfare, implying that strong Nash stable outcomes always exist and that the strong price of stability is $1$. We further prove that the strong price of anarchy is exactly $2$. In particular, we are able to show that, even for jointly cooperative deviations of at most $2$ agents, the strong price of anarchy is at most $2$ (we emphasize that, as we will describe in the next paragraph, the price of anarchy for Nash stable outcomes that are resistant to deviations of one agent grows linearly with the number of agents), while it is at least $2$ for jointly cooperative deviations of any subsets of agents.

We subsequently turn our attention on Nash stable outcomes. We notice that Nash stable outcomes are guaranteed to exist only if edge weights are non-negative.
In a similar way as in \cite{BFFMM14}, we prove that the price of anarchy is at least $\Omega(n)$, where $n$ is the number of agents, even for unweighted paths, and that it is at most $n-1$ for the more general case of non-negative edge-weighted graphs, thus giving an asymptotically tight characterization.
We also prove a matching lower bound of $\Omega(n)$ to the price of stability.

We finally consider core stable outcomes and show that they always exist, and in particular that an outcome that is core stable can be computed in polynomial time, even in the presence of negative weights, i.e., for general undirected weighted graphs. We then establish that the core price of stability is $2$. We further show that the core price of anarchy is at most $4$. We also provide a tight analysis for unweighted graphs.

In the next subsection we emphasize the differences between MFHGs and fractional hedonic games.

\subsection{Main Differences between MFHGs and Fractional Hedonic Games}
Roughly speaking, we say that an outcome is a $k$-\emph{strong Nash equilibrium} if no subset of at most $k$ agents can jointly change their strategies in a way that all of the $k$ agents strictly improve their utility. It is easy to see that, for any $k,k' \geq 2$, such that $k' \geq k$, a $k'$-strong Nash equilibrium is also a $k$-strong Nash equilibrium. It is known that $2$-strong Nash stable outcomes are not guaranteed to exist for fractional hedonic games, even for unweighted graphs \cite{BFFMM15}. In this paper we show that for MFHGs played on unweighted graphs, $k$-strong Nash equilibrium always exists and can be computed in polynomial time, for any $1 \leq k \leq n$, where $n$ is the number of agents, and provide a tight analysis on the strong price of anarchy and stability. 

For both MFHGs and Fractional Hedonic Games, Nash stable outcomes (or equivalently $1$-strong Nash stable) are guaranteed to exist  \cite{BFFMM14} for positive weights, but not for negative ones; moreover, the price of stability grows linearly with the number of agents. For fractional hedonic games played on unweighted graphs, it is known \cite{BFFMM15} that the price of stability is greater than $1$ even for simple graphs and that computing an optimum is NP-hard. For MFHGs we show that it is possible to compute in polynomial time a (strong) Nash equilibrium that is also optimum.

Finally, it is known that the core can be empty even for fractional hedonic games played on unweighted graphs and that it is NP-hard deciding the existence \cite{BBS15}. 
In this paper we show that for MFHGs the core is not empty for any graphs (this result was also observed in \cite{ABBHOP17} for unweighted graphs), and that a core stable outcome can be computed in polynomial time. We further provide a tight and an almost tight analysis for the core price of stability and anarchy, respectively.          

\section{Preliminaries}\label{sec:preliminaries}
For an integer $k>0$, denote with $[k]$ the set $\{1,\ldots,k\}$.

We model a coalition formation game by means of a undirected graph. For an undirected edge-weighted graph $G=(N,E,w)$, denote with $n=|N|$ the number of its nodes. For the sake of convenience, we adopt the notation $(i,j)$ and $w_{i,j}$ to denote the edge $\{i,j\}\in E$ and its weight $w(\{i,j\})$, respectively. Say that $G$ is unweighted if $w_{i,j}=1$ for each $(i,j)\in E$. We denote by $\delta^{i}(G) = \sum_{j \in N: (i,j) \in E} w_{i,j}$, the sum of the weights of all the edges incident to $i$. Moreover, let $\delta^{i}_{max}(G)=\max_{j \in N: (i,j) \in E} w_{i,j}$ be the maximum edge-weight incident to $i$. We will omit to specify $(G)$ when clear from the context. Given a set of edges $X\subseteq E$, denote with $W(X)=\sum_{(i,j) \in X}w_{i,j}$ the total weight of edges in $X$. Given a subset of nodes $S \subseteq N$, $G_{S}=(S,E_{S})$ is
the subgraph of $G$ induced by the set $S$, i.e., $E_S = \{(i,j)\in
E : i,j \in S\}$.

Given an undirected edge-weighted graph $G=(N,E,w)$, the {\em modified fractional hedonic game} induced by $G$,
denoted as $\ga(G)$, is the game in which each node $i \in N$ is associated with an agent. We assume that agents are numbered from $1$ to $n$ and, for every $i \in [n]$, each agent chooses to join a certain {\em coalition} among $n$ candidate ones: the strategy of agent $i$ is an integer $j \in [n]$, meaning that agent $i$ is selecting candidate coalition $C_j$. A coalition structure (also called outcome or partition) is a partition of the set of agents into $n$ coalitions $ \outcome=\{C_1,C_2,\ldots,C_n\}$ such that $C_j \subseteq N$ for each $j\in [n]$, $\bigcup_{j\in [n]}C_j = N$ and $C_i \cap C_j = \emptyset$ for any $i,j\in [n]$ with $i\neq j$. Notice that, since the number of candidate coalitions is equal to the number of agents (nodes), some coalition may be empty. If $i \in C_j$, we say that agent $i$ is a member of the coalition $C_j$. We denote by $\outcome(i)$  the coalition in $\outcome$ of which agent $i$ is a member. In an outcome $\outcome$, the utility of agent $i$ is defined as $u_i(\outcome)=\sum_{j\in \outcome(i)} \frac{w_{i,j}}{|\outcome(i)|-1}$. We notice that, for any possible outcome $\outcome$, we have that $u_i(\outcome) \leq \delta^{i}_{max}$. 

Each agent chooses the coalition she belongs to with the aim of maximizing her utility. We denote by $(\outcome, i, j)$, the new coalition structure obtained from $\outcome$ by moving agent $i$ from $\outcome(i)$ to $C_j$; formally, $(\outcome, i, j) = \outcome \setminus \{\outcome(i),C_j\} \cup \{\outcome(i)\setminus \{i\}, C_j \cup \{i\}\}$. 
An agent \emph{deviates} if she changes the coalition she belongs to. Given an outcome $\outcome$, an \emph{improving move} (or simply a \emph{move}) for agent $i$ is a deviation to any coalition $C_j$ that strictly increases her utility, i.e., $u_i((\outcome, i, j)) > u_i(\outcome)$. 
Moreover, agent $i$ performs a \emph{best-response} in coalition $\outcome$ by choosing a coalition providing her the highest possible utility (notice that a best-response is also a move when there exists a coalition  $C_j$ such that  $u_i((\outcome, i, j)) > u_i(\outcome)$).
An agent is \emph{stable} if she cannot perform a move. An outcome is \emph{(pure) Nash stable} (or a \emph{Nash equilibrium}) if every agent is stable. An \emph{improving dynamics}, or simply a dynamics, is a sequence of moves, while a \emph{best-response dynamics} is a sequence of best-responses.
A game has the {\em finite improvement path property} if it does not admit an improvement dynamics of infinite length. Clearly, a game possessing the finite improvement path property always admits a Nash stable
outcome. We denote with $\pne(\ga(G))$ the set of Nash stable
outcomes of $\ga(G)$. 

An outcome $\outcome$ is a $k$-\emph{strong Nash equilibrium} if, for each $\outcome'$ obtained from $\outcome$, when a subset of at most $k$ agents $K \subseteq N$ (with $|K| \leq k$) jointly change (or deviate from) their strategies (not necessarily selecting the same candidate coalition), $u_i(\outcome) \geq u_i(\outcome')$ for some $i$ belonging to $K$, that is, after the joint collective deviation, there always exists an agent in the set of deviating ones who does not improve her utility. 
We denote with $\sne[k](\ga(G))$ the set of strong Nash stable
outcomes of $\ga(G)$. We simply say that an outcome $\outcome$ is a strong Nash equilibrium if $\outcome$ is an $n$-strong Nash equilibrium. It is easy to see that, for any graph $G$ and any $k\geq2$, $\sne[k](\ga(G)) \subseteq \sne[k-1](\ga(G))$, while the vice versa does not in general hold. Clearly, $\sne[1](\ga(G))=\pne(\ga(G))$.
Analogously to the notion of Nash equilibrium, also for strong Nash equilibria it is possible to define a dynamics as a sequence of improving moves, where each move performed by agents in $K$ leading from outcome $\outcome$ to outcome $\outcome'$ is such that all of them improve their utility, i.e. $u_i(\outcome') > u_i(\outcome)$ for every $i \in K$.

We say that a coalition $T \subseteq N$ \emph{strongly blocks} an outcome $\outcome$, if each agent $i \in T$ strictly prefers $T$, i.e., strictly improve her utility with respect to her current coalition $\outcome(i)$.
An outcome that does not admit a strongly blocking coalition is called \emph{core stable} and is said to be in the \emph{core}. We denote with $\core(\ga(G))$ the core of $\ga(G)$. 


The \emph{social welfare} of a coalition structure $\outcome$ is the summation of the agents' utilities, i.e.,
$\SW(\outcome) = \sum_{i \in N} u_i(\outcome)$. We overload the social
welfare function by applying it also to single coalitions to obtain
their contribution to the social welfare, i.e., for any $i \in [n]$, $\SW(C_i)=\sum_{j
\in C_i} u_j(\outcome)$ so that $\SW(\outcome) = \sum_{i \in [n]} \SW(C_i)$. It is worth noticing that, equivalently, for any $i \in [n]$, $\SW(C_i)=\frac{2W\left(E_{C_i}\right)}{|C_i|-1}$ and $\SW(\outcome) =\sum_{i \in [n]}\frac{2 W\left(E_{C_i}\right) }{|C_i|-1}$. 

Given a game $\ga(G)$, an \emph{optimum} coalition structure $\outcome^*(\ga(G))$ is one that maximizes the social welfare of $\ga(G)$. 
The {\em price of anarchy} (resp. \emph{strong price of anarchy} and \emph{core price of anarchy}) of a modified fractional hedonic game $\ga(G)$ is defined as the worst-case ratio between the social welfare of a social optimum outcome and that of a Nash equilibrium (resp. strong Nash equilibrium and core). Formally, for any $k=1,\dots,n$, $\poa({\ga(G)}) = \max_{\outcome\in{\pne(\ga(G))}} \frac{\SW(\outcome^*(\ga(G)))}{\SW(\outcome)}$ (resp. $\spoa[k]({\ga(G)}) = \max_{\outcome\in{\sne[k](\ga(G))}} \frac{\SW(\outcome^*(\ga(G)))}{\SW(\outcome)}$ and $\cpoa({\ga(G)}) = \max_{\outcome\in{\core(\ga(G))}} \frac{\SW(\outcome^*(\ga(G)))}{\SW(\outcome)}$). Analogously, the {\em price of stability} (resp. \emph{strong price of stability} and \emph{core price of stability}) of $\ga(G)$ is defined as the best-case ratio between the social welfare of a social optimum outcome and that of a Nash equilibrium (resp. strong Nash equilibrium and core). Formally, for any $k=1,\dots,n$, $\pos({\ga(G)}) = \min_{\outcome\in{\pne(\ga(G))}} \frac{\SW(\outcome^*(\ga(G)))}{\SW(\outcome)}$ (resp. $\spos[k]({\ga(G)}) = \min_{\outcome\in{\sne[k](\ga(G))}} \frac{\SW(\outcome^*(\ga(G)))}{\SW(\outcome)}$ and $\cpos({\ga(G)}) = \min_{\outcome\in{\core(\ga(G))}} \frac{\SW(\outcome^*(\ga(G)))}{\SW(\outcome)}$). Clearly, for any game $\ga(G)$ it holds that $1 \leq \pos({\ga(G)}) \leq \poa({\ga(G)})$ (resp. $1 \leq \spos[k]({\ga(G)}) \leq \spoa[k]({\ga(G)})$ and $1 \leq \cpos({\ga(G)}) \leq \cpoa({\ga(G)})$).  


%
%

%
%

\section{Strong Nash stable outcomes}\label{sec:strong_Nash_equilibria}
%

In this section we consider strong Nash stable outcomes. We start by showing that even the existence of $2$-strong nash equilibria is not guaranteed for non-negative edge-weights graphs.

\begin{theorem}\label{thm:no_existence_strong_NAsh}
There exists a star graph $G$ containing only non-negative edge-weights such that $\{\ga(G)\}$ admits no $2$-strong Nash stable outcome. 
\end{theorem}  

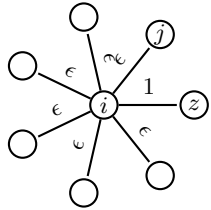
\begin{figure}[h]
\centering
 \scalebox{1}{\begin{tikzpicture}[font=\footnotesize,-,>=stealth',shorten >=1pt,auto,node distance=3cm, thick,main node/.style={circle,draw}]
\tikzset{edge/.style = {-,> = latex'}}

\def \radius {1.2cm}

\node (0) at (0, 0) {$i$};
  \node[main node] (0) at (0,0) {};

\node (1) at ({360/7 *(1 - 1)}:\radius) {$z$};
 \node[main node](1) at ({360/7 * (1 - 1)}:\radius) {};
\draw[edge] (0)  to node[sloped, anchor=center, above]  {$1$} (1);

\node (2) at ({360/7 * (2 - 1)}:\radius) {$j$};
 \node[main node](2) at ({360/7 * (2 - 1)}:\radius) {};
\draw[edge] (0)  to node[sloped, anchor=center, above]  {$\epsilon$} (2);

\node (3) at ({360/7 * (3 - 1)}:\radius) {};
 \node[main node](3) at ({360/7 * (3 - 1)}:\radius) {};
\draw[edge] (0)  to node[sloped, anchor=center, above]  {$\epsilon$} (3);

\node (4) at ({360/7 * (4 - 1)}:\radius) {};
 \node[main node](4) at ({360/7 * (4 - 1)}:\radius) {};
\draw[edge] (0)  to node[sloped, anchor=center, above]  {$\epsilon$} (4);

\node (5) at ({360/7 * (5 - 1)}:\radius) {};
 \node[main node](5) at ({360/7 * (5 - 1)}:\radius) {};
\draw[edge] (0)  to node[sloped, anchor=center, above]  {$\epsilon$} (5);

\node (6) at ({360/7 * (6 - 1)}:\radius) {};
 \node[main node](6) at ({360/7 * (6 - 1)}:\radius) {};
\draw[edge] (0)  to node[sloped, anchor=center, above]  {$\epsilon$} (6);

\node (7) at ({360/7 * (7 - 1)}:\radius) {};
 \node[main node](7) at ({360/7 * (7 - 1)}:\radius) {};
\draw[edge] (0)  to node[sloped, anchor=center, above]  {$\epsilon$} (7);
\end{tikzpicture}
}
  \caption{The star graph $G$.}\label{fig:star}
\end{figure}
\begin{proof}
Let $G$ be a star of order $n$ centred in $i$ as depicted in Figure~\ref{fig:star}. The weights of the edges are such that there exists a node leaf $z$ such that $w_{i,z}=1$, while for all the other leafs we have that, $w_{i,j}=\epsilon$, for any $j \neq z,i$, and for small enough positive $\epsilon>0$, (for instance $\epsilon < \frac{1}{n}$). First notice that, the grand coalition where all the agents belong to the same coalition, is not a $2$-strong Nash stable outcome since, for instance, the two agents $i, z$ would both get strictly higher utility if they belong to a different coalition containing only them. On the other hand, any outcome where any leaf $j$ does not belong to the same coalition containing the center $i$ is even not Nash stable (i.e., $1$-strong Nash stable), since $j$ would get utility zero but she can improve her utility by selecting the coalition containing the agent $i$. Hence, the claim follows.
\end{proof}

Given the above negative result, in the remainder of this section, we focus on unweighted graphs. 

Let $K_1$, $K_2$ and $K_3$ be the unweighted cliques with $1$, $2$ and $3$ nodes, respectively, i.e., $K_1$ is an isolated node,  $K_2$ has $2$ nodes and a unique edge and $K_3$ is a triangle with $3$ edges. We say that a coalition being isomorphic to $K_1$, $K_2$ or $K_3$ is a \emph{basic} coalition.

\subsection{Strong Price of Stability}
In this subsection we show that, for unweighted graphs, it is possible to compute in polynomial time an optimum outcome and also a strong Nash outcome with the same social value. As consequence we get that the strong price of stability is $1$.

In order to show how to compute in polynomial time an optimal solution, we first need some additional lemmata.

\begin{lemma}\label{lem:induction}
Given a coalition $C$ with $|C|\geq 4$, there exists an edge $e=(i,j)$ belonging to $E_C$ such that 
$$\SW(\{i,j\})+ \SW(C \setminus \{i,j\}) \geq \SW(C).$$

\end{lemma}
\begin{proof}
Let $m=|E_C|$ and $k=|C|$ be the number of edges and nodes in coalition $C$, respectively. 
Moreover, let $e=(i, j)$ be the edge minimizing $\Delta=\delta^i+\delta^j$.
Let us assume by contradiction that 
$$\SW(\{i,j\})+ \SW(C \setminus \{i,j\})=2+\frac{2(m-\Delta+1)}{k-3}< \frac{2m}{k-1}=\SW(C).$$

By simple calculations, we obtain that
\begin{align}\label{eq:delta}
\Delta > \frac{k^2-3k+2+2m}{k-1}
\end{align}
We denote by $\delta_{max}$ and $\delta_{min}$ the maximum and the minimum degrees of nodes in $G_C$, respectively. We have 
\begin{align}
2m&=\sum_{i\in C} \delta_i\geq(k-1)\delta_{min}+\delta_{max}\label{eq:edges_1}\\
\Delta &\leq \delta_{max}+\delta_{min}\label{eq:edges_2}
\end{align}
Substituting~\eqref{eq:edges_1}, ~\eqref{eq:edges_2}, in~\eqref{eq:delta}, the following holds: 
\begin{align*}\label{eq:diseq}
\Delta > \frac{k^2-3k+2+2m}{k-1}&\geq \frac{k^2-3k+2+(k-1)\delta_{min}+\delta_{max}}{k-1}\\
\delta_{max}+\delta_{min}\geq\Delta &> \frac{k^2-3k+2+(k-1)\delta_{min}+\delta_{max}}{k-1}\\
\left(\delta_{max}+\delta_{min}\right)(k-1)&> k^2-3k+2+(k-1)\delta_{min}+\delta_{max}\\
k\delta_{max}-\delta_{max}&>k^2-3k+2+\delta_{max}\\
(k-2)\delta_{max}&>(k-1)(k-2)\\
\delta_{max}&>(k-1):
\end{align*}
a contradiction, because the maximum degree of a node is at most $k-1$. 
\end{proof}

We are now ready to prove the following theorem, showing that it is possible to consider, without decreasing the social welfare of the outcome, only coalition structures formed by \emph{basic} coalitions.

\begin{theorem}\label{thm:basic_coalitions}
For any coalition structure $\outcome$, there exists a coalition structure $\outcome'$ containing only \emph{basic} coalitions and such that $\SW(\outcome')\geq \SW(\outcome)$.
\end{theorem}
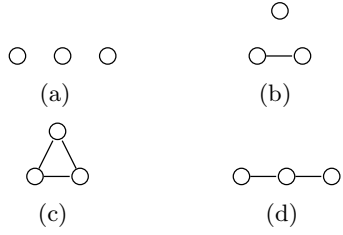
\begin{figure}[t]
\centering
\hspace{-10mm}\subfigure[]{\scalebox{.6}{  \begin{tikzpicture}[font=\footnotesize,-,>=stealth',shorten >=1pt,auto,node distance=3cm,
  thick,main node/.style={circle,draw}, node/.style={sloped,anchor=south,auto=false}]
  \tikzset{edge/.style = {- = latex'}}

\node (1) at (0, 0) {};
  \node[main node] (1) at (0,0) {};
\node (2) at (1, 0) {};
  \node[main node] (2) at (1,0) {};
\node (3) at (2, 0) {};
  \node[main node] (3) at (2, 0) {};

\end{tikzpicture}}}~
\hspace{15mm}\subfigure[]{\scalebox{.6}{\begin{tikzpicture}[font=\footnotesize,-,>=stealth',shorten >=1pt,auto,node distance=3cm,
  thick,main node/.style={circle,draw}, node/.style={sloped,anchor=south,auto=false}]
  \tikzset{edge/.style = {- = latex'}}

\node (1) at (0, 0) {};
  \node[main node] (1) at (0,0) {};
\node (2) at (1, 0) {};
  \node[main node] (2) at (1,0) {};
\node (3) at (0.5, 1) {};
  \node[main node] (3) at (0.5, 1) {};

\draw[edge] (1) to (2);

\end{tikzpicture}}}\\
\hspace{-3mm}\subfigure[]{\scalebox{.6}{  \begin{tikzpicture}[font=\footnotesize,-,>=stealth',shorten >=1pt,auto,node distance=3cm,
  thick,main node/.style={circle,draw}, node/.style={sloped,anchor=south,auto=false}]
  \tikzset{edge/.style = {- = latex'}}

\node (1) at (0, 0) {};
  \node[main node] (1) at (0,0) {};
\node (2) at (1, 0) {};
  \node[main node] (2) at (1,0) {};
\node (3) at (0.5, 1) {};
  \node[main node] (3) at (0.5, 1) {};

\draw[edge] (1) to (2);
\draw[edge] (2) to (3);
\draw[edge] (1) to (3); 
\end{tikzpicture}
}}~
\hspace{15mm}\subfigure[]{\scalebox{.6}{  \begin{tikzpicture}[font=\footnotesize,-,>=stealth',shorten >=1pt,auto,node distance=3cm,
  thick,main node/.style={circle,draw}, node/.style={sloped,anchor=south,auto=false}]
  \tikzset{edge/.style = {- = latex'}}

\node (1) at (0, 0) {};
  \node[main node] (1) at (0,0) {};
\node (2) at (1, 0) {};
  \node[main node] (2) at (1,0) {};
\node (3) at (2, 0) {};
  \node[main node] (3) at (2, 0) {};

\draw[edge] (1) to (2);
\draw[edge] (2) to (3);
\end{tikzpicture}
}}\\
\caption{Possible coalitions with three nodes.}
\label{fig:coalition}
\end{figure}
\begin{proof}
Consider any coalition $C$ belonging to $\outcome$. 
In the following we show that either coalition $C$ is \emph{basic}, or the nodes in $C$ can be partitioned in $h\geq 2$ basic coalitions $C'_1,\dots,C'_h$ such that $\sum_{i=1}^h \SW(C'_i) \geq \SW(C)$. This statement proves the claim because we can consider and sum up over all coalitions $C$ belonging to $\outcome$.

We prove the statement by induction on the number $k$ of nodes in $C$.

The base of the induction is for $k\leq 3$: 
For $k=1$ and $k=2$, $C$ is already a basic coalition.
For $k=3$, there are four possible configurations shown in Figure~\ref{fig:coalition}.
For configurations (a), (b) and (c), again $C$ already is a basic coalition (or can be trivially divided in basic coalitions).
For configuration (d), let $x_1, x_2, x_3$ the $3$ nodes in $C$; clearly, $\SW(C)=2$. Consider coalitions $C'_1=\{x_1,x_2\}$ and $C'_2=\{x_3\}$. It is easy to check that $\SW(C'_1)+\SW(C'_2)=2=\SW(\outcome)$.  

As to the induction step, given any $k\geq 4$, assume now that the statement holds for $1, \dots, k-1$; we want to show that it also holds for $k$. 

By Lemma \ref{lem:induction}, we know that there exists an edge $e=(i,j)$ belonging to $E_C$ such that $\SW(\{i,j\})+ \SW(C \setminus \{i,j\}) \geq \SW(C)$. 
Since $|C \setminus \{i,j\}|\leq k-2$, by the induction hypothesis,  coalition $C \setminus \{i,j\}$ can be decomposed in $h$ basic coalitions $C''_1,\dots,C''_h$ such that $\sum_{i=1}^h \SW(C''_i) \geq \SW(C \setminus \{i,j\})$. Therefore, given that also $\{i,j\}$ is a basic coalition, we have proven the induction step.
\end{proof}

By Theorem \ref{thm:basic_coalitions}, in order to compute an optimal solution for the coalition structure generation problem (i.e., an outcome maximizing the social welfare), it is possible to exploit a result from \cite{Hell84}:

\begin{theorem}[\cite{Hell84}]\label{thm:hell}
Given an unweighted graph $G$, it is possible to compute in polynomial time a partition of the nodes of $G$ in sets inducing subgraphs isomorphic to $K_1$, $K_2$ or $K_3$ (i.e., a coalition structure composed by basic coalitions) maximazing the number of nodes belonging to sets inducing subgraphs isomorphic to $K_2$ or $K_3$.
\end{theorem}

In fact, by combining Theorems \ref{thm:basic_coalitions} and \ref{thm:hell}, it is possible to prove the following result.

\begin{theorem}\label{thm:unweightedSocialOpt}
Given an unweighted graph $G$, there exists a polynomial time algorithm for computing a coalition structure $\outcome^*$ maximizing the social welfare.
\end{theorem}
\begin{proof}
By Theorem \ref{thm:basic_coalitions}, there must exist an optimal outcome $\outcome^*=(C^*_1,\dots,C^*_n)$ in which, for all $i=1,\dots,n$, $C^*_i$ is a basic coalition. Notice that any node in a basic coalition isomorphic to $K_1$ does not contribute to the social welfare, while all nodes in other coalitions contribute $1$ to $\SW(\outcome^*)$.
It follows that, in order to maximize the social welfare, the number of nodes belonging to coalitions isomorphic to $K_2$ or $K_3$ has to be maximized, and therefore the solution computed in Theorem \ref{thm:hell} is optimal also for our problem.
\end{proof}

In \cite{KKP16} the authors show that the price of stability of modified unweighted fractional hedonic games is $1$, without considering complexity issues. 
The different characterization of the optimum done in Theorem \ref{thm:basic_coalitions} allows us to first compute in polynomial time an outcome that maximizes the social welfare (done in Theorem \ref{thm:unweightedSocialOpt}) and then to transform this optimal outcome into a strong Nash without worsening its social welfare, again by a polynomial time transformation.
The following theorem completes this picture by providing a polynomial time algorithm for transforming an optical outcome into a strong Nash with the same social welfare, thus also proving that the strong price of stability is $1$.

\begin{theorem}\label{thm:calcoloStrongNash}
Given an unweighted graph $G$, 
it is possible to compute in polynomial time
an outcome $\outcome \in \sne[n]$ and such that $\SW(\outcome)=\SW(\outcome^*)$.
\end{theorem}
\begin{proof}
Let $\outcome^*$ be the optimal outcome computed in polynomial time by  Theorem \ref{thm:unweightedSocialOpt}.
Let $N'\subseteq N$ the set of agents belonging in $\outcome^*$ to coalitions isomorphic to $K_2$ or $K_3$. Notice that $\SW(\outcome^*) = |N'|$. No agent in $i \in N'$ can have an incentive in changing her strategy (and thus can belong to any deviating subset of agents), because $u_i(\outcome)=1$ and a node can gain at most $1$ in any outcome.
Therefore, if $N'=N$, then $\outcome^*$ is also a strong Nash equilibrium and the claim directly follows. 

In order to complete the proof, it is sufficient to (i) show the existence of  a dynamics involving only the set of agents $K \subseteq N''$, where $N''= N \setminus N'$, and leading to a strong Nash outcome $\outcome$; (ii) providing a polynomial time algorithm for computing $\outcome$.

For any $h=1,2,3$, let $\outcome^*_h \subseteq \outcome^*$ be the set containing all coalitions of $\outcome^*$ isomorphic to $K_h$. We first provide some useful properties of nodes in $N''$:
\begin{itemize}
\item[(P1)] For any couple of distinct nodes $i,j \in N''$, edge $(i,j) \not\in E$, because otherwise the social welfare of $\outcome^*$ could be improved by putting $i$ and $j$ in the same coalition: a contradiction to the optimality of $\outcome^*$.

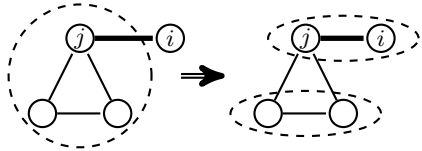
\begin{figure}[h]
\centering

 \scalebox{1}{
   \begin{tikzpicture}[font=\footnotesize,-,>=stealth',shorten >=1pt,auto,node distance=3cm,
  thick,main node/.style={circle,draw}, node/.style={sloped,anchor=south,auto=false}]
  \tikzset{edge/.style = {- = latex'}}

\node (1) at (0, 0) {};
  \node[main node] (1) at (0,0) {};
\node (2) at (1, 0) {};
  \node[main node] (2) at (1,0) {};
\node (3) at (0.5, 1) {$j$};
  \node[main node] (3) at (0.5, 1) {};
\node (4) at (1.7, 1) {$i$};
  \node[main node] (4a) at (1.7, 1) {};

\draw[edge] (1) to (2);
\draw[edge] (2) to (3);
\draw[edge] (1) to (3); 
\draw[edge, line width=2] (3) to (4);

\node (a) at (1.7, 0.5) {};
\node (b) at (2.6, 0.5) {};

\draw[dashed] (0.5,0.5) circle (0.95cm);

\draw[double,->] (a) to (b) {};

\node (1a) at (3, 0) {};
  \node[main node] (1a) at (3,0) {};
\node (2a) at (4, 0) {};
  \node[main node] (2a) at (4,0) {};
\node (3a) at (3.5, 1) {$j$};
  \node[main node] (3a) at (3.5, 1) {};
\node (4a) at (4.5, 1) {$i$};
  \node[main node] (4a) at (4.5, 1) {};

\draw[edge] (1a) to (2a);
\draw[edge] (2a) to (3a);
\draw[edge] (1a) to (3a); 
\draw[edge, line width=2] (3a) to (4a);

\draw[dashed] (3.5,0) ellipse (1cm and 0.3cm);
\draw[dashed] (4, 1) ellipse (1cm and 0.3cm);

\end{tikzpicture}
}
  \caption{Proof of (P2).}\label{fig:prop2}
\end{figure}

\item[(P2)] For any $i \in N''$ and any vertex $j$ belonging to a coalition in $\outcome^*_3$, edge $(i,j) \not\in E$, because otherwise the social welfare of $\outcome^*$ could be improved by removing $j$ from her current coalition and putting her in the same coalition of $i$: a contradiction to the optimality of $\outcome^*$ ( see Figure~\ref{fig:prop2}).

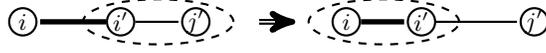
\begin{figure}[h]
\centering
 \scalebox{1}{
   \begin{tikzpicture}[font=\footnotesize,-,>=stealth',shorten >=1pt,auto,node distance=3cm,
  thick,main node/.style={circle,draw}, node/.style={sloped,anchor=south,auto=false}]
  \tikzset{edge/.style = {- = latex'}}

\node (0) at (-1.3, 0) {$i$};
  \node[main node] (0) at (-1.3,0) {};
\node (1) at (0, 0) {$i'$};
  \node[main node] (1) at (0,0) {};
\node (2) at (1, 0) {$j'$};
  \node[main node] (2) at (1,0) {};

\draw[edge] (1) to (2); 
\draw[edge, line width=2] (1) to (0); 

\node (a) at (1.7, 0) {};
\node (b) at (2.5, 0) {};

\draw[dashed] (0.5,0) ellipse (1cm and 0.3cm);

\draw[double,->] (a) to (b) {};

\node (1a) at (3, 0) {$i$};
  \node[main node] (1a) at (3,0) {};
\node (2a) at (4, 0) {$i'$};
  \node[main node] (2a) at (4,0) {};
\node (3a) at (5.5, 0) {$j'$};
  \node[main node] (3a) at (5.5, 0) {};

\draw[edge, line width=2] (1a) to (2a);
\draw[edge] (2a) to (3a);

\draw[dashed] (3.5,0) ellipse (1cm and 0.3cm);

\end{tikzpicture}
}
  \caption{Proof of (P3).}\label{fig:prop3}
\end{figure}

\item[(P3)] For any couple of distinct nodes $i,j \in N''$ and any coalition $\{i',j'\} \in \outcome^*_2$, if there exists an edge connecting node $i$ to a node in $\{i',j'\}$ (assume without loss of generality to node $i'$, i.e. assume that $(i,i') \in E$), then  edge $(j,j') \not\in E$, because otherwise the social welfare of $\outcome^*$ could be improved by removing $i'$ and $j'$ from their current coalition and putting them in the same coalition of $i$ and $j$, respectively: a contradiction to the optimality of $\outcome^*$ (see Figure~\ref{fig:prop3}).

\begin{figure}[h]
\centering
 \scalebox{1}{
   \begin{tikzpicture}[font=\footnotesize,-,>=stealth',shorten >=1pt,auto,node distance=3cm,
  thick,main node/.style={circle,draw}, node/.style={sloped,anchor=south,auto=false}]
  \tikzset{edge/.style = {- = latex'}}

\node (0) at (0, -1.5) {$i$};
  \node[main node] (0) at (0,-1.5) {};
\node (1) at (0, 0) {$i'$};
  \node[main node] (1) at (0,0) {};
\node (2) at (1, 0) {$j'$};
  \node[main node] (2) at (1,0) {};

\draw[edge] (1) to (2); 
\draw[edge] (1) to (0);

\node (1b) at (2.5, 0) {$j''$};
  \node[main node] (1b) at (2.5,0) {};
\node (2b) at (3.5, 0) {$i''$};
  \node[main node] (2b) at (3.5,0) {};
\node (3b) at (3.5, -1.5) {$j$};
  \node[main node] (3b) at (3.5,-1.5) {};

\draw[edge] (1b) to (2b); 
\draw[edge] (2b) to (3b); 
\draw[edge, line width=2] (2) to (1b);

\node (a) at (4, -1) {};
\node (b) at (5, -1) {};

\draw[dashed] (0.5,0) ellipse (1cm and 0.3cm);
\draw[dashed] (3,0) ellipse (1cm and 0.3cm);

\draw[double,->] (a) to (b) {};

%
%
%
%

\end{tikzpicture}
} 
 \scalebox{1}{  \begin{tikzpicture}[font=\footnotesize,-,>=stealth',shorten >=1pt,auto,node distance=3cm,
  thick,main node/.style={circle,draw}, node/.style={sloped,anchor=south,auto=false}]
  \tikzset{edge/.style = {- = latex'}}

\node (0) at (0, -1.5) {$i$};
  \node[main node] (0) at (0,-1.5) {};
\node (1) at (0, 0) {$i'$};
  \node[main node] (1) at (0,0) {};
\node (2) at (1, 0) {$j'$};
  \node[main node] (2) at (1,0) {};

\draw[edge] (1) to (2); 
\draw[edge] (1) to (0);

\node (1b) at (2.5, 0) {$j''$};
  \node[main node] (1b) at (2.5,0) {};
\node (2b) at (3.5, 0) {$i''$};
  \node[main node] (2b) at (3.5,0) {};
\node (3b) at (3.5, -1.5) {$j$};
  \node[main node] (3b) at (3.5,-1.5) {};

\draw[edge] (1b) to (2b); 
\draw[edge] (2b) to (3b); 
\draw[edge, line width=2] (2) to (1b);

\draw[dashed] (0,-0.75) ellipse (0.4cm and 1.1cm);
\draw[dashed] (1.75,0) ellipse (1.1cm and 0.3cm);
\draw[dashed] (3.5,-0.75) ellipse (0.4cm and 1.1cm);

\end{tikzpicture}
}
  \caption{Proof of (P4).}\label{fig:prop4}
\end{figure}
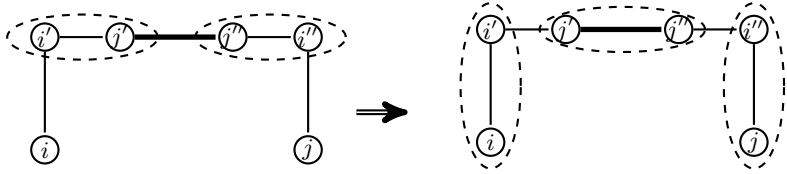
\item[(P4)] For any couple of distinct nodes $i,j \in N''$ and any couple of coalitions $\{i',j'\},\{i'',j''\}\in\outcome^*_2$, if there exist an edge connecting node $i$ to a node in $\{i',j'\}$ (assume without loss of generality to node $i'$, i.e. assume that $(i,i') \in E$), and another edge connecting node $j$ to a node in $\{i'',j''\}$ (assume without loss of generality to node $i''$, i.e. assume that $(j,i'') \in E$), then  edge $(j',j'') \not\in E$, because otherwise the social welfare of $\outcome^*$ could be improved by removing $i'$, $i''$ and $j'$ from their current coalition and putting them in the same coalition of $i$, $j$ and $j''$,  respectively: a contradiction to the optimality of $\outcome^*$ (see Figure~\ref{fig:prop4}).
\end{itemize}

Consider an \emph{initial} dynamics, ending in outcome $\outcome^0$, in which every agent in $i \in N''$ unilaterally moves in order to increase her utility (that in $\outcome^*$ is $0$). By properties (P1) and (P2) it follows that, for any $i \in N''$, $i$ selects a coalition in $\outcome^*_2$ and by property (P3) it follows that after this \emph{initial} dynamics, all coalitions in $\outcome^0 \setminus \outcome^*$ (i.e., all coalitions modified by this initial dynamics) are isomorphic to star graphs, i.e. only one node has degree greater than $1$. 

Consider now a sequence of improving moves performed by any subset of agents $K \subseteq N$ and such that for any $i \in K$, agent $i$ improves her utility after this move. 
For any $t \geq 1$, let $\outcome^t$ be the outcome reached after the $t$-th move of this dynamics and $K^t$ be the set of moving agents.
We want to show that this dynamics converges, i.e., that a strong Nash equilibrium is eventually reached. 

By properties (P3) and (P4) it follows that:
\begin{itemize}
\item[(P5)] For any coalition in $\outcome^*_2$, there exists an agent that will always have utility $1$ during any dynamics; let $\bar{N} \subseteq N$ the set containing these nodes. 
Clearly, every agent in $\bar{N}$, as well as all nodes belonging to coalitions in $\outcome^*_3$, will never belong to a subset of nodes performing an improving move and therefore will always remain in the same coalition she belongs in $\outcome^*$. 

\item[(P6)] For any $t\geq 1$, and any agent $i \in K^t$ (potentially $i$ could be an agent of a coalition in $\outcome^*_1$ or also an agent of a coalition in $\outcome^*_2$ not belonging to $\bar{N}$), $\outcome^t(i)$ is such that there exists a unique $j \in \outcome^t(i) \cap \bar{N}$ and $i$ will have a unique edge in $\outcome^t(i)$ connecting her to $j$.
\end{itemize}

By properties (P5) and (P6), the only nodes participating in the dynamics  are nodes either belonging to coalitions in $\outcome^*_1$ or belonging to coalitions in $\outcome^*_2$ but not belonging to $\bar{N}$; let $\bbar(N) $ be the set of these nodes, i.e., for any $t>1$, $K^t \subseteq \bbar(N)$.



In order to obtain a strong Nash equilibrium, we notice that the ``residual'' game played by agents in $\bbar(N)$ is equivalent to a \emph{singleton congestion game with identical latency functions} (CGI), in which we also have a set of resources (i.e. a strong Nash equilibrium in this new game is also a strong Nash equilibrium in our game and vice versa). 
In a CGI, agent's strategy consists of a resource. The delay of a resource is given by the number of agents choosing it, and the cost that each agent aims at minimizing is the delay of her selected resource.
In particular, the set of agents is $\bbar(N)$ and the set of resources is $\bar{N}$. In fact, in our ``residual'' game every agent aims at minimizing the cardinality of the star coalition she belongs to.
In \cite{HHKS13} it has been shown how to compute in polynomial time a strong Nash equilibrium for a class of congestion games including the one of CGI.

Let us call $\outcome$ the obtained strong Nash equilibrium. It remains to show that $\SW(\outcome)=\SW(\outcome^*)$. Observe that the difference between $\outcome$ and $\outcome^*$ is that some coalitions belonging to $\outcome^*$ isomorphic to $K_2$ becomes a coalition isomorphic to a star graph in $\outcome$, and that some coalitions belonging to $\outcome^*$ isomorphic to $K_1$ disappears in $\outcome$. The claim follows by noticing that the contribution to the social welfare of a coalition isomorphic to $K_1$ is zero, and that the contribution to the social welfare of a coalition isomorphic to $K_2$ (whose value is $2$) is the same as the one of a coalition isomorphic to a star graph.
\end{proof}

As a direct consequence of Theorem \ref{thm:calcoloStrongNash}, the following corollary holds.
\begin{corollary}\label{cor:spos1}
For any unweighted graph $G$ and any $k=1,\dots,n$, $\spos[k](\ga(G))=1$.
\end{corollary}

\subsection{Strong Price of Anarchy}
In this subsection we study the strong price of anarchy for unweighted graphs. 

\begin{theorem}\label{thm:strong_price_of_anarchy_at_least_2}
Given any $\epsilon>0$, there exists an unweighted graph $G$ such that $\spoa[n](\ga(G))\geq 2-\epsilon$.
\end{theorem}

\begin{proof}
\begin{figure}[h]
\centering
 \scalebox{.8}{  \begin{tikzpicture}[font=\footnotesize,-,>=stealth',shorten >=1pt,auto,node distance=3cm,
  thick,main node/.style={circle,draw}]
  \tikzset{edge/.style = {-,> = latex'}}

\node (1) at (2.5, 0) {};
  \node[main node] (1) at (2.5,0) {};
\node (2) at (1, 1) {};
  \node[main node] (2) at (1,1) {};
\node (3) at (2, 1) {};
  \node[main node] (3) at (2,1) {};
\node (4) at (3, 1) {$\dots$};
\node (5) at (4, 1) {};
  \node[main node] (5) at (4,1) {};

\node (6) at (1, 2) {};
  \node[main node] (6) at (1,2) {};
\node (7) at (2, 2) {};
  \node[main node] (7) at (2,2) {};
\node (8) at (3, 2) {$\dots$};
\node (9) at (4, 2) {};
  \node[main node] (9) at (4,2) {};

\draw[edge] (1) to (2);
\draw[edge] (1) to (3);
\draw[edge] (1) to (5);

\draw[edge] (6) to (2);
\draw[edge] (7) to (3);
\draw[edge] (9) to (5);

\draw[bend right, edge] (6) to (7);
\draw[bend left, edge] (6) to (9);
\draw[bend right, edge] (9) to (7);
\draw[dashed] (2.5, 2.1) ellipse (2cm and 0.5cm);
\end{tikzpicture}
}
  \caption{The graph $G$.}\label{fig:clique}
\end{figure}
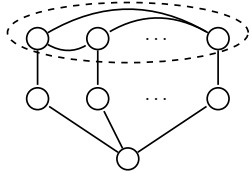
Let us consider the graph $G$ depicted in Figure~\ref{fig:clique}. The number of nodes in $G$ is $n=2k+1$. Specifically, we have $k$ agents $\{1,\ldots,k\}$ in the first (upper) layer, and other $k$ agents $\{k+1,\ldots,2k\}$ in the second layer. Moreover, the $k$ nodes in the upper layer form a clique. It is easy to see that the optimum solution $OPT$ has social welfare at least $\SW(OPT)\geq 2k$. In fact, the coalitions structure composed by $k$ non-empty coalitions corresponding to the $k$ matchings between agents of the first and second layer, i.e., for any $j=1,\ldots,k$, $C_j=\{j,k+j\}$ has social welfare exactly $2k$. A strong Nash stable outcome is given by the coalition structure $\outcome$ composed by two coalitions $\outcome=\{C_1, C_2\}$, where $C_1$ contains all the agents of the clique, while $C_2$ contains all the other agents. Indeed, on the one hand, all the agents belonging to the coalition $C_1$ get utility $1$ that is the maximum one they can get, which means that they do not have any interest on deviating from $\outcome$. Therefore, suppose by contradiction that $\outcome$ is not strong Nash stable, then the set of deviating agents must be a subset of the agents belonging to the coalition $C_2$. However, by using the fact that the two non-empty coalitions $C_1$ and $C_2$ of $\outcome$ contains the same number of agents, it is easy to see any subset of agents of $C_2$ cannot jointly deviate and all get higher utility with respect to $\outcome$. It follows that $\outcome$ is a strong Nash stable outcome. Since $\SW(\outcome)=k+2$, it follows that $\spoa[n] \geq \frac{2k}{k+2}$.
\end{proof}

\begin{theorem}\label{thm:strong_price_of_anarchy_at_most_2}
For any unweighted graph $G$, $\spoa[2](\ga(G))\leq 2$.
\end{theorem}
\begin{proof}
Let $\outcome^*$ the optimal solution computed by Theorem \ref{thm:unweightedSocialOpt}, in which all coalitions are basic ones.

Consider any $2$-strong Nash equilibrium $\outcome$.

For any coalition $C^*=\{i,j\}$ of $\outcome^*$ isomorphic to $K_2$, on the one hand we have that $\SW(C^*)=2$. On the other hand, since $\outcome$ is
 a $2$-strong Nash stable outcome, $u_i(\outcome)=1$ or $u_j(\outcome)=1$, because otherwise $i$ and $j$ could jointly perform an improving move.
  Thus, $u_i(\outcome) + u_j(\outcome)\geq 1$, whereas $u_i(\outcome^*) + u_j(\outcome^*) =2$. 

For any coalition $C^*=\{i,j,k\}$ of $\outcome^*$ isomorphic to $K_3$, on the one hand we have that $\SW(C^*)=3$. On the other hand, since $\outcome$ is a $2$-strong Nash stable outcome, at least $2$ agents among $i,j,k$ must have utility $1$ in $\outcome$, because otherwise there would exist two agents aiming at jointly perform an improving move: a contradiction to the $2$-strong Nash stability of $\outcome$. Thus, $u_i(\outcome) + u_j(\outcome)+ u_k(\outcome) \geq 2$, whereas $u_i(\outcome^*) + u_j(\outcome^*)+ u_k(\outcome^*) =3$.

For any $h=1,2,3$, let $N_h \subseteq N$ be such that for any $j \in N_h$, $C^*_j$is isomorphic to $K_h$.
Since agents being in coalitions of $\outcome^*$ isomorphic to $K_1$ do not contribute to $\SW(\outcome^*)$, we obtain 
\begin{eqnarray*}
\frac{\SW(\outcome^*)}{\SW(\outcome)}& \leq &
\frac{\sum_{j \in N_2} \SW(C_j^*)  + \sum_{j \in N_3} \SW(C^*_j)}
{\sum_{j \in N_2} \sum_{i \in C^*_j}{u_i(\outcome)} + \sum_{j \in N_3} \sum_{i \in C^*_j}{u_i(\outcome)}}\\ 
&\leq&  
\frac{\sum_{j \in N_2} \SW(C_j^*)  + \sum_{j \in N_3} \SW(C^*_j)}
{\sum_{j \in N_2} \frac{1}{2} \SW(C^*_j) + \sum_{j \in N_3} \frac{2}{3}\SW(C^*_j)}\\
&\leq&  
\frac{\sum_{j \in N_2} \SW(C_j^*)  + \sum_{j \in N_3} \SW(C^*_j)}
{\frac{1}{2}\left(\sum_{j \in N_2} \SW(C^*_j) + \sum_{j \in N_3} \SW(C^*_j)\right)}=2
\end{eqnarray*} \end{proof}

From Theorems \ref{thm:strong_price_of_anarchy_at_least_2} and \ref{thm:strong_price_of_anarchy_at_most_2}, we immediately get the following result.
\begin{corollary}
The strong price of anarchy for unweighted graphs is $2$.
\end{corollary}


\section{Nash stable outcomes}\label{sec:Nash_equilibria}
In this section we consider Nash stable outcomes. We start by showing that there exists a graph $G$ containing negative edge-weights such that the game induced by $G$ admits no Nash stable outcome. This result is very similar to Lemma $1$ of~\cite{BFFMM14}.

\begin{theorem}\label{thm:no_existence_Nash_negative_weight}
There exists a graph G containing edges with negative weights such that ${\ga(G)}$ admits no Nash stable outcome. 
\end{theorem} 

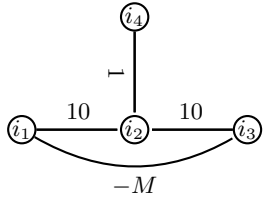
\begin{figure}
\centering
 \scalebox{1}{ \begin{tikzpicture}[font=\footnotesize,-,>=stealth',shorten >=1pt,auto,node distance=3cm,
  thick,main node/.style={circle,draw}]
  \tikzset{edge/.style = {-,> = latex'}}

\node (1) at (0, 0) {$i_1$};
  \node[main node] (1) at (0,0) {};
\node (2) at (1.5, 0) {$i_2$};
  \node[main node] (2) at (1.5, 0) {};
\node (3) at (3, 0) {$i_3$};
  \node[main node] (3) at (3,0) {};
\node (4) at (1.5, 1.5) {$i_4$};
  \node[main node] (4) at (1.5,1.5) {};

\draw[edge] (1) to (2);
\draw[edge] (3) to (2);
\draw[edge] (4) to (2); 
\draw[edge] (1) to node[sloped, anchor=center, above] {$10$} (2);
\draw[edge] (3) to node[sloped, anchor=center, above] {$10$} (2);
\draw[edge] (4) to node[sloped, anchor=center, below] {$1$} (2);
\draw[bend right, edge] (1) to node[sloped, anchor=center, below] {$-M$} (3);
\end{tikzpicture}}
  \caption{The graph $G$.}\label{fig:eq}
\end{figure}
\begin{proof}
Let $G$ be the graph in Figure~\ref{fig:eq} and 
fix a Nash stable outcome $\outcome$.
It is easy to see that, for $-M$ small enough, agents $i_1$ and $i_3$ cannot belong to the same coalition. By contrast, agents $i_4$ and $i_2$ must belong to the same coalition since otherwise the utility of $i_4$ would be zero. Let $C_j$ be the coalition containing agents $i_4$ and $i_2$. If $C_j=\{i_2, i_4\}$, then agent $i_1$ wants to join the coalition and improve her utility from zero to $10/2=5$ thus contradicting
the fact that $\outcome$ is Nash stable. If $C_j\supset\{i_2 , i_4\}$, then, since agents $i_1$ and $i_3$ cannot belong to the same coalition, it must be $|C_j|= 3$. Moreover, there exists a coalition $C_i$ containing exactly one between the two agents $i_1$ and $i_3$. Hence, we get
the utility of agent $i_2$ in $\outcome$ is $11/2< 10$, while $10$ is the utility in joining coalition $C_i$, which rises again a contradiction. Since all possibilities for $C_j$ have
been considered, it follows that a Nash stable outcome cannot exist.
\end{proof}

We further show that there exists a dynamic of infinite length for games played on unweighted graphs.  

\begin{theorem}\label{thm:noconvergenceNash}
There exists an unweighted graph $G$ such that ${\ga(G)}$ does not possess the finite improvement path property, even under best-response dynamics. 
\end{theorem} 
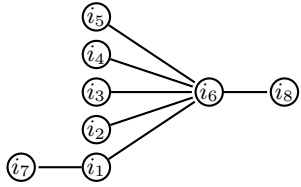
\begin{figure}
\centering
 \scalebox{1}{  \begin{tikzpicture}[font=\footnotesize,-,>=stealth',shorten >=1pt,auto,node distance=3cm,
  thick,main node/.style={circle,draw}]
  \tikzset{edge/.style = {-,> = latex'}}

\node (1) at (0, 0) {$i_1$};
  \node[main node] (1) at (0,0) {};
\node (2) at (0, 0.5) {$i_2$};
  \node[main node] (2) at (0,0.5) {};
\node (3) at (0, 1) {$i_3$};
  \node[main node] (3) at (0,1) {};
\node (4) at (0, 1.5) {$i_4$};
  \node[main node] (4) at (0,1.5) {};
\node (5) at (0, 2) {$i_5$};
  \node[main node] (5) at (0,2) {};

\node (6) at (1.5, 1) {$i_6$};
  \node[main node] (6) at (1.5,1) {};
\node (7) at (2.5, 1) {$i_8$};
  \node[main node] (7) at (2.5,1) {};
  
\node (8) at (-1, 0) {$i_7$};
  \node[main node] (8) at (-1,0) {};
  
\draw[edge] (1) to (8);
\draw[edge] (6) to (7);
\draw[edge] (1) to (6);
\draw[edge] (2) to (6);
\draw[edge] (3) to (6);
\draw[edge] (4) to (6);
\draw[edge] (5) to (6);
  \end{tikzpicture}
}
  \caption{The graph $G$.}\label{fig:dynamic}
\end{figure}
\begin{proof}
Let us consider the game induced by be the graph $G$ depicted in Figure~\ref{fig:dynamic}. Let us analyze the dynamics that starts from the coalitions structure $\outcome=\{\{i_1,\dots,i_7\}, \{i_8\}\}$, where agents $\{i_1,\dots,i_7\}$ are together in a coalition, and agent $i_8$ is alone in another one. It is not difficult to check that, if the agents perform their unique (best) improving moves in the following exact ordering $i_6,i_1,i_7,i_2,i_3,i_4,i_6,i_1,i_7,i_4,i_3,i_2$, we get back to the starting coalitions structure $\outcome$.    
\end{proof}

Despite the above negative results, it is easy to see that, if a graph $G$ does not contain negative edge-weights, then the game induced by $G$ admits a Nash equilibrium, that is the outcome where all the agents are in the same coalition. Therefore, in the next subsections we characterize the efficiency of Nash stable outcomes in modified fractional hedonic games played on
general graphs with non-negative edge-weights.

%
%

By definition, we have that $1\leq PoS\leq PoA$.

\subsection{Price of Anarchy}

We first show that the price of anarchy grows linearly with the number of agents, even for the special case of unweighted paths.

\begin{theorem}\label{thm:lb_poa_Nash_unweighted_graphs}
There exists an unweighted path $G$ such that $\poa(\ga(G)) = \Omega(n)$.
\end{theorem}

\begin{proof}
Let $G$ be an unweighted simple path with an even number $n$ of nodes. Notice that, since in this setting the utility of an agent in any outcome is at most $1$, the optimum solution is given by a perfect matching, that is,  $\SW(OPT) = n$. However, when all the nodes are in the same coalition, we obtain a Nash stable outcome $\outcome$ such that $\SW(\outcome)= \frac{2*(n-2)+2}{n-1}$. Hence, the claim follows.
\end{proof}


We are able to show an asymptotically matching upper bound, holding for weighted (positive) graphs.
\begin{theorem}\label{thm:ub_poa_Nash_unweighted_graphs}
For any weighted graph with non-negative edge-weights $G$, $\poa(\ga(G))\leq n-1$.
\end{theorem}

\begin{proof}
We notice that in any Nash equilibrium $\outcome$, any agent $i$ has utility $u_i(\outcome) \geq \frac{\delta^{i}_{max}}{n-1}$, since agent $i$ can always join the coalition containing the agent $j$, where $j=\text{arg}\max_{z \in N} w_{i,z}$. On the other hand, in the optimal outcome $OPT$, we have that any agent $i$ has utility such that $u_i(OPT)\leq \frac{(|OPT(i)|-1)*\delta^{i}_{max}}{|OPT(i)|-1} = \delta^{i}_{max}$. Hence, by summing over all agents, the theorem follows. 
\end{proof} 


\subsection{Price of Stability}
On the one hand, since we have proved in Corollary \ref{cor:spos1} that, for the setting of unweighted graphs, the strong price of stability is $1$, it directly follows that also the price of stability is $1$ in this setting, because any strong Nash equilibrium is also a Nash equilibrium.

On the other hand, in the weighted case, given the upper bound to the price of anarchy provided in Theorem \ref{thm:ub_poa_Nash_unweighted_graphs}, the following theorem shows an asymptotically matching lower bound to the price of stability. 

\begin{theorem}
There exists a weighted star $G$ with non-negative edge weights such that $\pos(\ga(G))= \Omega(n)$.
\end{theorem}

\begin{figure}[h]
\centering
 \scalebox{1}{\begin{tikzpicture}[font=\footnotesize,-,>=stealth',shorten >=1pt,auto,node distance=3cm, thick,main node/.style={circle,draw}]
\tikzset{edge/.style = {-,> = latex'}}

\def \radius {1.2cm}

\node (0) at (0, 0) {$i$};
  \node[main node] (0) at (0,0) {};

\node (1) at ({360/7 *(1 - 1)}:\radius) {$z$};
 \node[main node](1) at ({360/7 * (1 - 1)}:\radius) {};
\draw[edge] (0)  to node[sloped, anchor=center, above]  {$1$} (1);

\node (2) at ({360/7 * (2 - 1)}:\radius) {$j$};
 \node[main node](2) at ({360/7 * (2 - 1)}:\radius) {};
\draw[edge] (0)  to node[sloped, anchor=center, above]  {$\epsilon$} (2);

\node (3) at ({360/7 * (3 - 1)}:\radius) {};
 \node[main node](3) at ({360/7 * (3 - 1)}:\radius) {};
\draw[edge] (0)  to node[sloped, anchor=center, above]  {$\epsilon$} (3);

\node (4) at ({360/7 * (4 - 1)}:\radius) {};
 \node[main node](4) at ({360/7 * (4 - 1)}:\radius) {};
\draw[edge] (0)  to node[sloped, anchor=center, above]  {$\epsilon$} (4);

\node (5) at ({360/7 * (5 - 1)}:\radius) {};
 \node[main node](5) at ({360/7 * (5 - 1)}:\radius) {};
\draw[edge] (0)  to node[sloped, anchor=center, above]  {$\epsilon$} (5);

\node (6) at ({360/7 * (6 - 1)}:\radius) {};
 \node[main node](6) at ({360/7 * (6 - 1)}:\radius) {};
\draw[edge] (0)  to node[sloped, anchor=center, above]  {$\epsilon$} (6);

\node (7) at ({360/7 * (7 - 1)}:\radius) {};
 \node[main node](7) at ({360/7 * (7 - 1)}:\radius) {};
\draw[edge] (0)  to node[sloped, anchor=center, above]  {$\epsilon$} (7);
\end{tikzpicture}
}
  \caption{The star graph $G$.}\label{fig:star_2}
\end{figure}
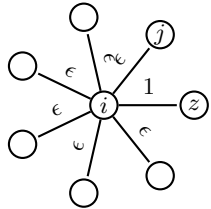
\begin{proof}
Let $G$ be a star with $n$ nodes centred in $i$ as depicted in Figure~\ref{fig:star_2}. The weights of the edges are such that there exists a leaf node $z$ such that $w_{i,z}=1$, while for all the other leaf nodes $j \neq z,i$ we have $w_{i,j}=\epsilon$, for an arbitrarily small $\epsilon>0$, (for instance $0 < \epsilon < \frac{1}{n}$). 

Notice that the grand coalition (i.e., the outcome in which all agents belong to the same coalition) is the unique Nash stable outcome and has social welfare equal to $\frac{2+2(n-2)\epsilon}{n-1}$. 
In fact, in any Nash equilibrium, all the leafs must be in the coalition together with the center $i$. 
On the other hand, the coalition containing only agents $i$ and $z$ yields a social value of $2$, and thus the theorem follows.
\end{proof}


\section{Core stable outcomes}\label{sec:core_equilibria}

We first show that the strict core of $\ga(G)$ could be empty, even if $G$ is unweighted.

\begin{theorem}\label{thm:strictcore_empty}
There exists an unweighted graph $G$ such that $\score(\ga(G))=\emptyset$.
\end{theorem}
\begin{proof}
Let $G$ be a path with $n=3$ nodes $\{i, j, k\}$. 

If $\outcome=\{\{x_1\},\{ x_2\}, \{x_3\}\}$, $C=\{x_1,x_2\}$ is a blocking coalition. In fact, moving from their coalition in $\outcome$ to coalition $C$, both $x_2$ and $x_3$ increase their utility form $0$ to $1$.

If $\outcome=\{\{x_1, x_2\}, \{x_3\}\}$, $C=\{x_2,x_3\}$ is a weakly blocking coalition. In fact, moving from their coalition in $\outcome$ to coalition $C$, $x_3$ increases her utility form $0$ to $1$ and $x_2$ does not change her utility. The case $\outcome=\{\{ x_1\}, \{x_2, x_3\}\}$ is symmetric.

Finally, if $\outcome=\{\{x_1, x_2, x_3\}\}$, $C=\{x_1,x_2\}$ is a weakly blocking coalition. In fact, moving from their coalition in $\outcome$ to coalition $C$, $x_1$ increases her utility form $\frac{1}{2}$ to $1$ and $x_2$ does not change her utility.

Since all possibilities for $\outcome$ have been considered, it follows that a strict core stable coalition does not exist.  
\end{proof}

Given the negative result of Theorem \ref{thm:strictcore_empty} concerning the strict core of modified fractional hedonic games, in the following we focus on the core of this games.

In this section we consider the core of MFHGs. 
We first show that for any graph $G$, the core of the game $\ga(G)$ in not empty, and that a core stable outcome approximating the optimal social welfare by a factor of $2$ can be computed in polynomial time.

\begin{theorem}\label{thm:2-apx_weighted_K2}
Given any graph $G=(N,E,w)$, there exists a polynomial time algorithm for computing a core stable coalition structure $\outcome$ such that $\SW(\outcome) \geq \frac{1}{2} \SW(\outcome^*(\ga(G)))$ and all coalitions in $\outcome$ are of cardinality at most $2$.
\end{theorem}
\begin{proof}
Consider the following algorithm, working in phases $t=1,2,\dots$.
Let $G^0=(N,E^0,w)$ be the subgraph of $G$ such that $E^0=\{e \in E: w(e)\geq 0\}$, that is, $G^0$ has the same vertices as $G$ and only contains the edges of $G$ of non-negative weight. For any $t\geq1$, let $G^t=(N^t, E^t, w)$ be the graph obtained after phase $t$. In any phase $t \geq 1$, a new coalition isomorphic to $K_2$ is added to $\outcome$ as follows:
Let $e^{t-1}=\{i,j\}$ be an edge in $E^{t-1}$ of maximum weight $w_{i,j}=\max_{e \in E^{t-1}}w_e$.
We add to $\outcome$ the coalition formed by $i$ and $j$, i.e., $\outcome = \outcome \cup \{i,j\}$.  
Moreover, let $G^t$ such that $N^t = N^{t-1} \setminus \{i,j\}$ and $E^t \subset E^{t-1}$ the subset of edges of $G^0$ induced by nodes $N^t$.

When $E^t=\emptyset$, the algorithm ends returning $\outcome \cup \{ \{i\}|i \in N^t\}$.
Since at each phase at least an edge is removed from the graph, the algorithms terminates in at most $|E|$ phases returning an outcome with all coalitions of cardinality at most $2$.

We first show that $\outcome$ is a core stable outcome of $\ga(G)$. Remind that, for any possible outcome, $u_i(\outcome) \leq \delta^{i}_{max}$. Therefore, in the outcome $\outcome$, agents $i$ and $j$ selected at phase $t=1$ are achieving the maximum utility they can hope. It implies that such agents cannot belong to any strongly block coalition. The proof continues by induction as follows. Suppose that all the agents selected until phase $q$, i.e., agents belonging to $N \setminus N^q$, cannot belong to any strongly block coalition, then agents $i_{q+1}$ and $j_{q+1}$ selected in the phase $q+1$ cannot belong to any strongly block coalition as well. In fact, suppose that such agents have a certain utility $x$ in the coalition $\outcome$. For the inductive hypothesis we have that they can create a strongly block coalition only with agents belonging to $N^{q+1}$. However, since the edge $(i_{q+1},j_{q+1})$ has the maximum weights in $G^{q+1}$, if implies that they cannot get utility greater than $x$. Finally, for the agents that are not matched, i.e., agents that are alone in a coalition, since they form and independent set, they cannot form a strongly block coalition, and this finishes the proof.

It remains to show that $\SW(\outcome) \geq \frac{1}{2} \SW(\outcome^*(\ga(G)))$.
First of all notice that in any phase $t$, a coalition contributing $2 w_{e^{t-1}}$ to the social welfare is added to $\outcome$; we thus obtain that
$$ \SW(\outcome) = \sum_{t \geq 1}  2 w_{e^{t-1}}.$$ 

For any $e \in E$, let $f(e,i) \in \{0,1,2\}$ be the number of endpoints of $e$ belonging to coalition $C^*_i$. 
It is possible to bound $\SW(\outcome^*(\ga(G)))$ as follows:
\begin{eqnarray}
\SW(\outcome^*)& = &
\sum_{C^*_i \in \outcome^*} \SW(C^*_i)\nonumber\\
& = &
\sum_{C^*_i \in \outcome^*}
 \sum_{t \geq 1}
   \sum_{e \in E_{C^*_i} \cap (E^{t} \setminus E^{t-1})}
      \frac{2 w_e}
           {|C^*_i-1|}\nonumber\\
& \leq &
\sum_{C^*_i \in \outcome^*}
 \sum_{t \geq 1}
      \frac{2 f(e^{t-1},i) w_{e^{t-1}} (|C^*_i-1|)}
           {|C^*_i-1|}\label{passaggio1}\\
& = &
\sum_{t \geq 1}
 \sum_{C^*_i \in \outcome^*}
      2 f(e^{t-1},i) w_{e^{t-1}}\nonumber\\           
& = &
\sum_{t \geq 1}
   4 w_{e^{t-1}},\label{passaggio2}
\end{eqnarray} 
where inequality \ref{passaggio1} holds because $w_{e^{t-1}}=\max_{e \in E^{t-1}}w_e$ and every endpoint of $e^{t-1}$ belonging to $C^*_i$ can have at most $|C^*_i-1|$ adjacent edges (notice that all edges in $E^{t} \setminus E^{t-1}$ are adjacent to an endpoint of $e^{t-1}$), and equality \ref{passaggio2} holds because, given that $C^*_1,\ldots,C^*_n$ are a partition of $N$, it follows by definition of $f$ that $\sum_{C^*_i \in \outcome^*}f(e^{t-1},i)=2$. 
Therefore,
$$
\frac{\SW(\outcome^*(\ga(G)))}{\SW(\outcome)} \leq 
\frac{\sum_{t \geq 1}
   4 w_{e^{t-1}}}
{\sum_{t\geq 1} 2 w_{e^{t-1}}} = 2.$$ 
\end{proof} 

As a direct consequence of Theorem \ref{thm:2-apx_weighted_K2}, the following corollary holds.

\begin{corollary}
For any graph $G$, $\cpos(\ga(G)) \leq 2$.
\end{corollary}

We now show a matching lower bound on the $\cpos$ for the case of weighted graphs.
\begin{theorem}
For any $\epsilon>0$, there exists a weighted graph $G$ such that $\cpos(\ga(G)) \geq 2-\epsilon$.
\end{theorem}
\begin{proof}
Consider the graph $G$ represented in Figure~\ref{fig:path}.

\begin{figure}[h]
\centering
 \scalebox{1}{  \begin{tikzpicture}[font=\footnotesize,-,>=stealth',shorten >=1pt,auto,node distance=3cm,
  thick,main node/.style={circle,draw}]
  \tikzset{edge/.style = {-,> = latex'}}

\node (4) at (3.5, 0) {$i_1$};
  \node[main node] (4) at (3.5,0) {};
\node (5) at (4.5, 0) {$i_2$};
  \node[main node] (5) at (4.5,0) {};
\node (6) at (5.5, 0) {$i_3$};
  \node[main node] (6) at (5.5,0) {};
\node (7) at (6.5, 0) {$i_4$};
  \node[main node] (7) at (6.5,0) {};

\draw[edge] (4)  to node[sloped, anchor=center, above]  {$1$} (5);
\draw[edge] (5)  to node[sloped, anchor=center, above] {$1+\frac{\epsilon}{2}$} (6);  
\draw[edge] (6)  to node[sloped, anchor=center, above] {$1$}(7);

\end{tikzpicture}
}
  \caption{Graph $G$.}\label{fig:path}
\end{figure}
On the one hand, it is easy to check that the only core stable coalition $\outcome$ is the one where the two central agents $i_2$ and $i_3$ are together in the same coalition, while agent $i_1$, as well as agent $i_4$, are alone in different coalitions, i.e.,  $\outcome=\{\{i_1\},\{i_2, i_3\},\{i_4\}\}$. Notice that $\SW(\outcome)=2\left(1+\frac{\epsilon}{2}\right)$. On the other hand, the outcome $\outcome'=\{\{i_1, i_2\},\{i_3, i_4\}\}$, has a social welfare equal to $4$, and therefore $\SW(\outcome^*)\geq 4$. It follows that $\cpos(\ga(G)) \geq \frac{4}{2\left(1+\frac{\epsilon}{2}\right)}\geq 2-\epsilon$.
\end{proof}

For unweighted graphs, it is easy to see that the optimum outcome produced in Theorem \ref{thm:calcoloStrongNash} is also core stable, and therefore the following proposition holds:

\begin{proposition}
For any unweighted graph $G$, $\cpos(\ga(G))=1$.
\end{proposition}

We are also able to prove a constant upper bound to the core price of anarchy.

\begin{theorem}
For any graph $G$, $\cpoa(\ga(G))\leq 4$.
\end{theorem}
\begin{proof}
Let $\outcome'$ be the solution computed by Theorem \ref{thm:2-apx_weighted_K2}, in which all coalitions have cardinality at most $2$.

Consider any core stable outcome $\outcome$.

For any coalition $C'=\{i,j\}$ of $\outcome'$ isomorphic to $K_2$, on the one hand we have that $\SW(C')=2$. On the other hand, since $\outcome$ is
 a core stable outcome, $u_i(\outcome)=1$ or $u_j(\outcome)=1$, because otherwise coalition $\{i,j\}$ would strongly block outcome $\outcome$.
Thus, $u_i(\outcome) + u_j(\outcome)\geq 1$, whereas $u_i(\outcome') + u_j(\outcome') =2$. 

Let $N' \subseteq N$ be such that for any $j \in N'$, $C'_j$ is isomorphic to $K_2$.
Since agents being in all other coalitions of $\outcome'$ do not contribute to $\SW(\outcome')$, we obtain 
\begin{eqnarray*}
\frac{\SW(\outcome')}{\SW(\outcome)}& \leq &
\frac{\sum_{j \in N'} \SW(C'_j)}
{\sum_{j \in N'} \sum_{i \in C'_j}{u_i(\outcome)}}\\ 
&\leq&  
\frac{\sum_{j \in N_2} \SW(C'_j) }
{\sum_{j \in N_2} \frac{1}{2} \SW(C'_j) }=2.
\end{eqnarray*} 
The claim follows because, by Lemma \ref{thm:2-apx_weighted_K2}, $\SW(\outcome^*(\ga(G))) \leq 2 \cdot \SW(\outcome')$.
\end{proof}

For unweighted graphs we get the following tight characterization on the core price of anarchy.

\begin{proposition}
For any unweighted graph $G$, $\cpoa(\ga(G))=2$.
\end{proposition}
\begin{proof}
For the lower bound, it is easy to see that, given an unweighted path of four nodes $i_1,i_2,i_3,i_4$, the outcome $\outcome=\{\{i_1\}, \{i_2,i_3\}, \{i_4\}\}$ is core stable and has social welfare $2$, while the optimum outcome $\outcome^*=\{\{i_1,i_2\},\{i_3,i_4\}\}$ has social welfare $4$.  
A matching upper bound can be obtained by exploiting the same arguments used in the proof of Theorem \ref{thm:strong_price_of_anarchy_at_most_2}.
\end{proof}

\section{Conclusions}\label{sec:conclusion}

We notice that one could consider \emph{relaxed} strong Nash stable and \emph{strict} core outcomes, where among the agents that cooperatively deviate, all of them do not worsen their utility, and at least one of them gets a strictly better utility. However, these stable outcomes do not exist even for very simple instances. In fact, if $G$ is an unweighted path of $3$ nodes, $(\ga(G))$ admits no relaxed strong Nash stable outcomes as well as no strict core outcomes. 

There are some open problems suggested by our work. First of all, it would be nice to close the gap between the lower bound of $2$ for the core price of stability and the upper bound of $4$ for the core price of anarchy, and to study the complexity of computing an optimal outcome when the graph is weighted.
Another research direction could be that of designing truthful mechanisms for MFHGs that perform well with respect to the sum of the agents' utility.
Finally, it would be interesting to adopt different social welfare than the one considered in this paper. An example could be that of maximizing the minimum utility among the agents.


\newpage

\bibliographystyle{plainurl}  
\bibliography{bibliography}  

\end{document}